\documentclass[12pt, a4paper,reqno]{amsart}

\usepackage[T1]{fontenc}
\usepackage[utf8]{inputenc}

\usepackage{color} 
\definecolor{bleu_sombre}{rgb}{0,0,0.6}  \definecolor{rouge_sombre}{rgb}{0.8,0,0}\definecolor{vert_sombre}{rgb}{0,0.6,0}
\usepackage[plainpages=false,colorlinks,linkcolor=bleu_sombre,
citecolor=rouge_sombre,urlcolor=vert_sombre,breaklinks]{hyperref}

\usepackage[english]{babel}
\usepackage{vmargin}

\usepackage{amsmath,amssymb,amsthm,graphicx,amsfonts,url,color,enumerate,dsfont,stmaryrd, mathrsfs, csquotes}

\theoremstyle{plain}
\newtheorem{theorem}{{Theorem}}[section]
\newtheorem*{theorem*}{{Theorem}}
\newtheorem{proposition}[theorem]{Proposition}
\newtheorem{conjecture}[theorem]{Conjecture}
\newtheorem*{proposition*}{Proposition}
\newtheorem{corollary}[theorem]{Corollary}
\newtheorem*{corollary*}{Corollary}
\newtheorem{lemma}[theorem]{Lemma}

\newtheorem*{lemma*}{Lemma}
\theoremstyle{definition}
\newtheorem{definition}[theorem]{Definition}
\newtheorem*{definition*}{Definition}
\theoremstyle{remark}
\newtheorem{remark}[theorem]{Remark}

\makeatletter

\@addtoreset{equation}{section}
\makeatother

\renewcommand{\leq}{\leqslant}	\renewcommand{\geq}{\geqslant}

\newcommand{\R}{\mathbb{R}}	
\newcommand{\C}{\mathbb{C}}
\newcommand{\N}{\mathbb{N}}	

\newcommand{\dx}{\mathrm{d}}	
\newcommand{\eps}{\varepsilon}
\newcommand{\Dom}{\mathrm{Dom}\,}

\renewcommand{\Re}{\mathrm{Re}\,}
\renewcommand{\Im}{\mathrm{Im}\,}

\begin{document}

\title[]{Magnetic quantum currents\\ in the presence of a Neumann wall}

\author[N. Raymond]{Nicolas Raymond}
\address[N. Raymond]{Univ Angers, CNRS, LAREMA, SFR MATHSTIC, F-49000 Angers, France}
\email{nicolas.raymond@univ-angers.fr}

\author[\'E. Soccorsi]{\'Eric Soccorsi}
\address[\'E. Soccorsi]{Aix-Marseille Univ, Universit\'e de Toulon, CNRS, CPT, Marseille, France.}
\email{eric.soccorsi@univ-amu.fr}

\subjclass[2010]{}

\thanks{}

\date{}

\begin{abstract}
We consider the Schrödinger operator with constant transverse magnetic field on a half-plane, endowed with Neumann boundary conditions. We study the low energy currents flowing along the boundary and we establish a Limiting Absorption Principle for the magnetic Neumann Laplacian under perturbation of an electric potential. 
\end{abstract}

\maketitle

\section{Context and motivation}

\subsection{Definition of the main operator}

In this paper, we consider the Hamiltonian with constant magnetic field on the half-plane $\R^2_+=\{(x,y)\in\R^2 : x>0\}$ and Neumann boundary condition:
\begin{equation}\label{eq.0}
\mathscr{L}_h:=D^2_x+(hD_y-x)^2\,,\quad D=-i\partial\,,\quad h>0\,,
\end{equation}
acting on the domain
\begin{multline*}
\Dom(\mathscr{L}_h)=\Big\{\psi\in L^2(\R^2_+) : D_x \psi\in L^2(\R^2_+)\,,(hD_y-x)\psi\in L^2(\R^2_+)\,,\\
\left(D^2_x+(hD_y-x)^2\right)\psi\in L^2(\R^2_+)\, \mbox{ and } \forall y\in\R\,, \partial_x\psi(0,y)=0 \Big\}\,.
\end{multline*}
This operator appears in many contexts. For instance, it plays a major role in the study of surface superconductivity, see \cite{FH10}. It also acquired a life of its own over the years, see \cite{Raymond}. The present paper is concerned with the semiclassical spectral analysis of the operator $\mathscr{L}_h$, that is to say that we mostly focus on spectral results for the operator $\mathscr{L}_h$ as $h\to 0$. It might seem surprising at first sight that $y$ is the only variable affected by the semiclassical parameter $h$ in this model, but this can be explained as follows. By performing the rescaling
\[x=h^{-\frac 12}\mathsf{x}\,,\quad y=h^{\frac 12}\mathsf{y}\,,\]
in \eqref{eq.0}, we notice that the operator $\mathscr{L}_h$ is unitarily equivalent to 
\[B^{-1}\left(D^2_{\mathsf{x}}+(D_{\mathsf{y}}-B\mathsf{x})^2\right)\,,\quad \mbox{where}\ B=h^{-1}\,,\]
which shows that the semiclassical limit corresponds to the \emph{large magnetic field regime}. Such partially semiclassical scaling has already been used by numerous authors in spectral theory of magnetic Schrödinger operators, as, e.g., in \cite{BHR16, BHR22} and in \cite{RR19}, where magnetic models \textit{à la} Iwatsuka (see \cite{I85}) are considered.
This paper is also fitted in the semiclassical framework in order to give a lighter formulation of the results and to use the power of the theory of pseudo-differential operators.

\subsection{A short bibliography} 
Magnetic quantum Hall currents have attracted a lot of attention from the mathematical community over the last decades, see, e.g., \cite{CHS02, DBP99, EJK01, FM03, FGW00, HS08b, HS08a, MMP99}, this list being non exhaustive. 
Quantum Hall devices can be modeled by a magnetic Laplacian describing the planar motion of an electron submitted to a constant transverse magnetic field. The electron is confined to unbounded regions of the plane by potential barriers or Dirichlet boundary conditions. The edge, be it the edge of the electrostatic confining potential or the Dirichlet boundary of the spatial domain, creates edge currents, whereas for Iwatsuka Hamiltonians (see \cite{I85}), quantum currents are generated by changes in the strength of the magnetic field, see, e.g., \cite{DHS14,HS15}. Some authors have also observed that the role of the "edge" could be played by a discontinuity of the magnetic field, see \cite{HPRS16} or \cite{AK20} where the dispersion curves are studied in detail.

In translationally invariant unbounded straight-edge geometries such as the half-plane or the infinite strip, the unperturbed quantum Hall Hamiltonian is fibered and edge currents occur at energies associated with fibers whose group velocity is non-zero, see \cite{CHS02, DBP99, FGW00, HS08b, HS08a, MMP99}. In this case, the existence of edge currents can be linked to the absolutely continuous spectral nature of the spectrum through the use of Mourre's positive commutator method, see, e.g. \cite{BHRS10, BRS08}, or also \cite{RR19} (where the absolute continuity is established by elementary means for Iwatsuka Hamiltonians). Similarly,  limiting absorption for the quantum Hall Hamiltonian in $\R_+^2$ at all energies except Landau levels, was derived in \cite{PS16} from the monotonicity of the dispersion relations.

Notice that quantum Hall currents in two-edge geometries such as the infinite strip, propagate in opposite directions
along the left and right edges of the band, respectively, see \cite{HS08a}. This can be understood from the fact that, unlike for the case of one-edge geometries, the dispersion functions of the corresponding quantum Hamiltonian are no longer monotonic, see \cite{HS08b, HS08a}. Moreover, these functions being ''symmetric'' about their minimum value, the net current flowing across any line orthogonal to the strip, is zero. 

This picture is quite reminiscent of the one of the present paper where none of the dispersion functions of the Neumann Laplacian in $\R_+^2$ is monotonic. Nevertheless, in contrast to the quantum Hall Hamiltonian in the band, the dispersion relations of the Neumann Laplacian are not symmetric about their minimum, which motivates for a closer look into its transport and spectral properties.

\subsection{Basic results on low energy currents}
In this section, we describe some elementary results about the low energy currents associated with the Hamiltonian $\mathscr{L}_h$. Due to the translational invariance of this operator in the $y$-direction, we shall see that its discrete spectrum is empty.

\subsubsection{Fiber decomposition}
Let us define the (unitary) semiclassical Fourier transform with respect to $y$ as
\[\mathscr{F}_h\varphi(x,\xi)=\frac{1}{\sqrt{2\pi h}}\int_{\R}e^{-i\xi y/h}\varphi(x,y)\dx y\,.\]
Then, $\mathscr{L}_h$ being invariant in the $y$-direction, we have the direct integral decomposition
\[\mathscr{F}_h\mathscr{L}_{h}\mathscr{F}_h^{-1}=\int_{\R}^\oplus \ell_{\xi}\,\dx\xi\,,\]
where for all $\xi \in \R$, $\ell_{\xi}$ is the Neumann realization in $L^2(\R_+)$ of the differential operator
\[D^2_{x}+(\xi-x)^2\,.\]
The operator $\ell_\xi$ is sometimes referred as the "de Gennes operator''. Since $x \mapsto (\xi-x)^2$ is unbounded as $x$ goes to infinity, $\ell_\xi$ has a compact resolvent and we denote by $\mu_{j}(\xi)$ its $j$-th eigenvalue, $j \in \N:=\{1,2,\ldots \}$. The well-known properties of this operator are recalled in the following proposition (see the original paper \cite{DH93}, and \cite[Section 3.2]{FH10}, \cite[Section 2.4]{Raymond}).
\begin{proposition}\label{prop.rappel}
	For all $j \in \N$, it holds true that:
	\begin{enumerate}[\rm i.]
		\item $\xi \mapsto \mu_{j}(\xi)$ is analytic on $\R$ and has a unique minimum denoted by $\xi_{j-1}$, which is non-degenerate.
		\item $\mu_{j}$ decreases on $(-\infty,\xi_{j-1})$ and increases on $(\xi_{j-1},+\infty)$. 
		\item $\Theta_{j-1}=\xi^2_{j-1}$, where $\Theta_{j-1}:=\mu_{j}(\xi_{j-1})$. 
		\item $\lim_{\xi\to+\infty}\mu_{j}(\xi)=2j-1$ and $\lim_{\xi\to-\infty}\mu_{j}(\xi)=+\infty$.
		\item $\Theta_{0}\in(0,1)$ and $\Theta_{j-1}\in\left(2j-3,2j-1\right)$ for $j\geq 2$.
	\end{enumerate}
\end{proposition}
Since the $\mu_j$'s are non-constant by Proposition \ref{prop.rappel}.iv, the spectrum of $\mathscr{L}_h$ is purely absolutely continuous (see \cite{I85}). Moreover we have 
\[\mathsf{sp}(\mathscr{L}_h)=[\Theta_0,+\infty)\,,\]
see \cite[Theorem XIII.85]{RS78}.

\subsubsection{Basic properties of low energy magnetic quantum currents}

Let us define the current operator in the direction $y$ by
\[\mathscr{J}_h:=[\mathscr{L}_{h},iy]=2h(hD_{y}-x)\,,\]
and the current operator with energy concentration in $I$ by
\[\mathscr{J}_{h,I}:=\mathds{1}_{I}(\mathscr{L}_{h}) \mathscr{J}_{h}\mathds{1}_{I}(\mathscr{L}_{h})\,.\]
The operator $\mathscr{J}_{h,I}$ is bounded on $E_{h,I}:=\mathsf{range}(\mathds{1}_{I}(\mathscr{L}_{h})) $ and symmetric.

Put
\[\lambda_{\min}(h,I):=\inf\mathsf{sp}\left(\mathscr{J}_{h,I}\right)\,,\quad\lambda_{\max}(h,I):=\sup\mathsf{sp}\left(\mathscr{J}_{h,I}\right).\]
We recall from the min-max principle that
\[\lambda_{\min}(h,I)=\inf_{   \substack{\psi\in E_{h,I} \\ \psi\neq 0}    }\frac{\langle\mathscr{J}_{h,I}\psi,\psi\rangle}{\|\psi\|^2_{L^2(\R^2_{+})}}\,,\quad\lambda_{\max}(h,I)=\sup_{   \substack{\psi\in E_{h,I} \\ \psi\neq 0}    }\frac{\langle\mathscr{J}_{h,I}\psi,\psi\rangle}{\|\psi\|^2_{L^2(\R^2_{+})}}\,.\]
Then, the spectral radius of $\mathscr{J}_{h,I}$, given by
\[\rho\left(\mathscr{J}_{h,I}\right)=\max\left(|\lambda_{\min}(h,I)|,|\lambda_{\max}(h,I)|\right)\,,\]
is the maximal strength of the (absolute value of the) current carried by the energy window $I$.

\begin{definition}\label{defi.semi-current}
	Let $e\in[\Theta_{0}, \Theta_{1})$ and consider the equation
	\begin{equation}\label{eq.nu1e}
	\mu_{1}(\xi)=e\,.
	\end{equation}
	\begin{enumerate}[\rm i.]
		\item When $e\in(\Theta_{0},1)$, the equation \eqref{eq.nu1e} has two solutions $\mu^{-1}_{-}(e)<\xi_{0}<\mu^{-1}_{+}(e)$.
		\item When $e=\Theta_{0}$, the equation \eqref{eq.nu1e} has only one solution $\mu^{-1}_{\pm}(\Theta_0)=\xi_{0}$.
		\item When $e\in(1,\Theta_{1})$,  the equation \eqref{eq.nu1e} has one solution $\mu^{-1}_-(e)<\xi_{0}$ and, by convention, we set $\mu^{-1}_+(e)=+\infty$.
	\end{enumerate}
	For all $e\in[\Theta_{0}, \Theta_{1})$, we define the \emph{algebraic current} as
	\begin{equation}\label{eq.j}
	c(e)=\mu'_{1}(\mu^{-1}_-(e))+\mu'_{1}(\mu^{-1}_+(e))\,,
	\end{equation}
	where we have set $\mu'_{1}(+\infty)=0$.
	
\end{definition}

In the sequel, it will be assumed at some stage of the analysis that the following conjecture (which is supported by numerical simulations carried out by M. P. Sundqvist with Mathematica) is verified.

\begin{conjecture}\label{conj.der-3}
	We have $\mu_{1}^{(3)}(\xi_{0})<0$.
\end{conjecture}

\begin{lemma}\label{lem.current-dir}
	Assume that Conjecture \ref{conj.der-3} is true. Then, there exists $e_{\star}\in(\Theta_{0},1)$ such that for all $e\in(\Theta_{0},e_{\star})$, we have $c(e)<0$.
	
\end{lemma}
For $e\in[\Theta_{0},\Theta_{1})$, we set
\[I_{\delta}=[e-\delta,e+\delta]\,.\]

\begin{proposition}\label{prop.current-dir}
	Assume that $\R_{+}\ni h\mapsto \delta(h)$ tends to $0$ as $h \downarrow 0$. 
	Then, we have
	\[ h^{-1}\lambda_{\min}(I_{\delta},h)=\mu'_{1}(\mu^{-1}_-(e))+\mathcal{O}(\delta(h))\,,\quad  h^{-1}\lambda_{\max}(I_{\delta},h)=\mu'_{1}(\mu^{-1}_+(e))+\mathcal{O}(\delta(h))\,.\]
	In particular, the spectral radius reads
	\[\rho\left( h^{-1}\mathscr{J}_{I_{\delta},h}\right)=\max\left(|\mu'_{1}(\mu^{-1}_-(e))|,|\mu'_{1}(\mu^{-1}_+(e))|\right)+o(1)\,.\]
	Assume in addition that Conjecture \ref{conj.der-3} holds. Then, for all $e\in(\Theta_{0},e_{\star})$ the current of maximal strength is negative provided $h$ is sufficiently small, i.e., we have $\rho\left(\mathcal{J}_{I_{\delta},h}\right)=|\lambda_{\min}(I_\delta,h)|$.
\end{proposition}

Moreover, the low energy states are localized in the vicinity of the Neumann boundary $x=0$, as can be seen from the following Agmon type estimate.

\begin{proposition}\label{prop.Agmon}
Put
	\[J_e:=(-\infty,e)\,,\quad e<1\,.\]
Then, 
for all $h>0$, all $K>0$ and all $\psi\in E_{J_{e},h}$, we have
	\[\int_{\R_{+}^2}e^{K x}|\psi(x,y)|^2\dx x\dx y\leq C\|\psi\|^2_{L^2(\R^2_{+})}\,,\]
	for some constant $C>0$ depending only on $e$ and $K$.
\end{proposition}

\subsection{Mourre estimate and applications}
When the perturbation $V$ is switched off, the presence of a positive edge current with energy concentration in $I$ is, by definition of the current operator $\mathscr{J}_{h,I}$, tied to the existence of a positive local commutator estimate (or Mourre estimate) for the operator $\mathscr{L}_{h}$ in the same energy interval. This motivates for a closer look into the problem of designing a Mourre inequality for the perturbed magnetic Laplacian
\[\mathscr{L}_{h,V}:=D^2_x+(hD_y-x)^2+h^\gamma V(x,y,h)\,,\]
for some $\gamma>0$ and some suitably smooth (in a sense that will be made precise further) bounded perturbation $V$. Such an estimate will prove useful for characterizing the spectrum of $\mathscr{L}_{h,V}$ in the corresponding energy interval. We are more specifically interested in energy intervals lying in the vicinity of $\Theta_0$, in which the unperturbed current $\mathscr{J}_{h,I}$ tends to vanish. For that purpose we introduce the following spectral window located above $\Theta_0$, whose size depends on $h>0$: 
\[I:=[e-\delta,e+\delta]\,,\qquad e:=\Theta_0+ah^{\alpha}\in(\Theta_0,1)\,, \quad \delta:=bh^\beta\, , \]
where
\[ 0\leq\alpha<1\,,\quad \beta>2\alpha\,,\quad 0<b<a\,.\]
When $\alpha=0$, we choose $a\in(0,1-\Theta_0)$.

\subsubsection{Statement of the Mourre estimate}
In view of designing a Mourre estimate, \emph{i.e.}, a positive commutator estimate for the perturbed operator $\mathscr{L}_{h,V}$ on the "sliding" energy window $I$, where the strength of current is very weak, we pick a smooth compactly supported function $f_h$ satisfying 
\[\forall\xi\in(\mu^{-1}_\mp(e\pm h^\alpha), \mu_\mp^{-1}(e\mp h^\alpha))\,,\quad f_h(\xi)=\mp 1\,,\]
where $\mu^{-1}_\pm(E)$ denotes the largest/smallest solution of the equation $\mu(\xi)=E$. We can take $f_h\in S_{\frac{\alpha}{2}}(1)$, where, as in \cite[Section 4.4]{Z13}, we write for any (order) function $g$,
\[S_\gamma(g):=\{\psi\in\mathscr{C}^\infty(\R) : \forall\alpha\in\N\,,\exists C_\alpha>0\,,\quad |\partial^\alpha \psi|\leq C_\alpha h^{-\gamma|\alpha|} g\}\,.\]
Then, the function $(y,\eta)\mapsto yf_h(\eta)$ being in $S_{\frac \alpha 2}(\langle y\rangle)$, we set 
\[\mathscr{A}_h:=yf_h(hD_y)+f_h(hD_y)y=\mathrm{Op}^w_h\left(yf_h(\eta)\right)\,,\]
where we recall that the Weyl quantization is defined by
\[\mathrm{Op}^w_h a\,\psi(x)=\frac{1}{2\pi h}\int_{\R^{2}}e^{i\eta(x-y)/h}a\left(\frac{x+y}{2},\eta\right)\psi(y)\mathrm{d}y\mathrm{d}\eta\,.\]

For the sake of notational simplicity, we drop the dependence on $h$ and write $f$ instead of $f_h$ in the sequel.

\begin{theorem}[Mourre estimate]\label{theo.Mourre}
	Assume that $V(x,\cdot,h)\in S(\langle y\rangle^{-1})$ uniformly in $(x,h)$ and pick $\gamma \geq \beta$ so large that $\gamma>1+2\alpha$. Then, there exist $\tilde{c}_0>0$ and $h_0>0$ such that for all $h\in(0,h_0)$, we have:
	\[\forall \phi\in\mathrm{range}\, \mathds{1}_I(\mathscr{L}_{h,V})\,,\quad 
	\langle[\mathscr{L}_{h,V},i\mathscr{A}_h]\phi,\phi \rangle\geq \tilde{c}_0 h^{1+\alpha}\|\phi\|^2\,.
	\]
\end{theorem}

\subsubsection{Application}
One of the main consequences of the above Mourre estimate is the following Limiting Absorption Principle (LAP) for the operator $\mathscr{L}_{h,V}$ in $I$, and subsequently the absolute continuity of its spectrum in $I$. This is detailled in Section \ref{sec.PAL}.

\begin{corollary}\label{cor.PAL}
Let $V$ and $\gamma$ be the same as in Theorem \ref{theo.Mourre}.
Then, there exist $h_0>0$ and $C>0$ such that for all $h\in(0,h_0)$ and all uniformly (w.r.t. $h$) bounded operator $\mathscr{C}_h$ such that $\mathscr{C}_h\mathscr{A}_h$ and $\mathscr{A}_h\mathscr{C}_h$ are also uniformly bounded, we have for all $z\in I\times\R\setminus\{0\}$,
\[\|\mathscr{C}_h(\mathscr{L}_{h,V}-z)^{-1}\mathscr{C}_h\|\leq  C h^{\min(-2+3\alpha,-3/2+\alpha,-1-\alpha)}\,.\]
In particular, it holds true that:
\begin{enumerate}[\rm (i)]
\item For all $z\in I\times \R\setminus\{0\}$,
\[\|\langle y\rangle^{-1}(\mathscr{L}_{h,V}-z)^{-1}\langle y\rangle^{-1}\|\leq Ch^{\min(-2+3\alpha,-3/2+\alpha,-1-\alpha)}\,.\]
\item The spectrum of $\mathscr{L}_{h,V}$ lying in $I$ is absolutely continuous.
\end{enumerate}
\end{corollary}
\begin{remark}
For all $\alpha\in(0,1)$, it follows from Corollary \ref{cor.PAL} that the spectrum of the operator $\mathscr{L}_{h,V}$ is purely absolutely continuous spectrum in $(\Theta_0+ah^\alpha,1)$.
\end{remark}

\subsection{Structure of the article}
The paper is organized as follows. In Section \ref{sec.basic} we present basic considerations on low energy magnetic quantum currents and we prove Propositions \ref{prop.current-dir} and \ref{prop.Agmon}. We explain in Section \ref{sec.PAL} how to derive the LAP by revisiting the Mourre theory and estimating carefully the involved constants (what is crucial since the spectral window of interest also depends on $h$). Finally, in Section \ref{sec.Mourre}, we give the proof of Theorem \ref{theo.Mourre} and we establish Corollary \ref{cor.PAL} with the aid of Section \ref{sec.PAL}.

\section{Basic properties of low energy currents}\label{sec.basic}

\subsection{Proof of Lemma \ref{lem.current-dir}}
We have
\[\mu_{1}(\xi)=\mu_{1}(\xi_{0})+\frac{\mu''_{1}(\xi_{0})}{2}(\xi-\xi_{0})^2+\frac{\mu^{(3)}_{1}(\xi_{0})}{6}(\xi-\xi_{0})^3+r(\xi-\xi_{0})\,,\]
from the analyticity of $\mu_1$ at $\xi_{0}$, where $r$ is an analytic function at $0$ such that $r(p ) \underset{p\to 0}{=}\mathcal{O}(p^4)$. Since $\mu''_{1}(\xi_{0})>0$, the Morse lemma then yields that the equation
\[\mu_{1}(\xi)=\mu_{1}(\xi_{0})+\eps\]
has two distinct solutions, provided that $\eps>0$ is sufficiently small. Putting $\xi=\xi_{0}+p$ in the above equation, this amounts to saying that there exist two solutions $p_{-}(\eps)<p_{+}(\eps)$ such that
\[p_{\pm}(\eps)=\pm\sqrt{\eps}\sqrt{\frac{2}{\mu''_{1}(\xi_{0})}}+u_{\pm}(\eps)\,,\]
where $u_{\pm}(\eps)=o(\sqrt{\eps})$. Namely, we get through standard computations that
\[u_{\pm}(\eps)=-\frac{\mu_{1}^{(3)}(\xi_{0})}{6\mu''_{1}(\xi_{0})}\eps+o(\eps)\,.\]
As a consequence we have
\[\mu'_{1}(\xi_{0}+p_{-}(\eps))+\mu'_{1}(\xi_{0}+p_{+}(\eps))=\frac{5\mu_{1}^{(3)}(\xi_{0})}{6\mu''_{1}(\xi_{0})}\eps+o(\eps)\,,\]
which is negative whenever $\eps$ is sufficiently small.

\subsection{Proof of Proposition \ref{prop.current-dir}}

\subsubsection{Preliminaries}

\begin{lemma}\label{lem.spec-loc}
	For all $\psi\in E_{h,I_{\delta}}$, it holds true that 
	\[\mathscr{F}_h\psi(x,\xi)=\mathds{1}_{I_{\delta}}(\mu_{1}(\xi))\langle\mathscr{F}_h\psi(\cdot,\xi),u_{1}(\cdot,\xi)\rangle_{L^2(\R_{+})} u_{1}(x,\xi)\,,\quad (x,\xi) \in \R_+^2\,, \]
	where $u_j(\cdot,\xi)$ denotes the normalized eigenfunction associated with $\mu_j(\xi)$.
\end{lemma}

\begin{proof}
	Since $\mathds{1}_{I_{\delta}}(\mathscr{L}_{h})\psi=\psi$, we have
	\[\mathscr{F}_h\mathds{1}_{I_{\delta}}(\mathscr{L}_{h})\mathscr{F}^{-1}_h\mathscr{F}_h\psi=\mathds{1}_{I_{\delta}}(\mathscr{F}_h\mathscr{L}_{h}\mathscr{F}^{-1}_h)\mathscr{F}_h\psi=\mathscr{F}_h\psi\,.\]
	Plugging this into the following decomposition of $\mathscr{F}_h\psi(\cdot,\xi)$ on the 
	Hilbertian basis $(u_{j}(\cdot,\xi))_{j\geq 1}$ of $L^2(\R+)$,
	\[\mathscr{F}_h\psi(x,\xi)=\sum_{j\geq 1}\langle\mathscr{F}_h\psi(\cdot,\xi),u_{j}(\cdot,\xi)\rangle_{L^2(\R_{+})} u_{j}(x,\xi)\,,\]
	and taking into account that
	$\mathds{1}_{I_{\delta}}(\mathscr{F}_h\mathscr{L}_{h}\mathscr{F}^{-1}_h)u_{j}(\cdot,\xi)=\mathds{1}_{I_{\delta}}(\mu_{j}(\xi))u_{j}(\cdot,\xi)$, 
	we get the desired result upon remembering that $I_{\delta}\cap (\Theta_{1},+\infty)=\emptyset$.
\end{proof}

\begin{lemma}\label{lem.pars}
	For all $\psi\in E_{h,I_{\delta}}$, we have
	\[\langle\mathscr{J}_{I_{\delta},h}\psi,\psi\rangle= h \int_{\R}\mathds{1}_{I_{\delta}}(\mu_{1}(\xi))\mu'_{1}( \xi)|\langle\mathscr{F}_h\psi(\cdot,\xi), u_{1}(\cdot,\xi)\rangle_{L^2(\R_{+})}|^2\dx \xi\,.\]
\end{lemma}

\begin{proof}
	In light of Lemma \ref{lem.spec-loc}, we have
	\[\langle\mathscr{J}_{I_{\delta},h}\psi,\psi\rangle=2 h\int_{\R_{+}}\int_{\R}\mathds{1}_{I_{\delta}}(\mu_{1}(\xi))(\xi-x)|u_{1}(x,\xi)|^2|\langle\mathscr{F}_h\psi(\cdot,\xi), u_{1}(\cdot,\xi)\rangle_{L^2(\R_{+})}|^2\dx x\dx \xi\,,\]
	by the Parseval formula. Lemma \ref{lem.pars} follows from this and
	the Feynman-Hellmann formula
	\[\mu'_{1}(\xi)=2\int_{\R_{+}}(\xi-x)|u_{1}(x,\xi)|^2\dx x\,,\]
	upon applying Fubini's theorem.
\end{proof}
Armed with Lemma \ref{lem.pars}, we are now in position to prove Proposition \ref{prop.current-dir}.

\subsubsection{Estimation of the current}
With reference to Lemma \ref{lem.pars}, we have for all $\psi\in E_{h,I_{\delta}}$,
	\[\langle\mathscr{J}_{I_{\delta},h}\psi,\psi\rangle= h\int_{\R}\mathds{1}_{I_{\delta}}(\mu_1(\xi))\mu_1'( \xi)|\langle\mathscr{F}_h\psi(\cdot,\xi), u_{1}(\cdot,\xi)\rangle_{L^2(\R_{+})}|^2\dx \xi\,.\]
Therefore, we infer from the analyticity of $\mu_1$ in $I_\delta$ that
\begin{multline*}
 h^{-1}\langle\mathscr{J}_{I_{\delta},h}\psi,\psi\rangle=\int_{\mu_1^{-1}(e-\delta,e+\delta)\cap\{\xi<\xi_{0}\}}\mu_1'(\mu^{-1}_-(e)) |\langle\mathscr{F}_h\psi(\cdot,\xi), u_{1}(\cdot,\xi)\rangle_{L^2(\R_{+})}|^2\dx \xi\\
+\int_{\mu_1^{-1}(e-\delta,e+\delta)\cap\{\xi>\xi_{0}\}}\mu_1'(\mu^{-1}_+(e))|\langle\mathscr{F}_h\psi(\cdot,\xi), u_{1}(\cdot,\xi)\rangle_{L^2(\R_{+})}|^2\dx \xi 
+\mathcal{O}(\delta )\|\psi\|^2\,.
\end{multline*}
Bearing in mind that $(\xi-\xi_0) \mu_1'(\xi) \geq 0$ for all $\xi \in \mu_1^{-1}(e-\delta,e+\delta)$, the min-max theorem then yields that
\[ h^{-1}\lambda_{\min}(\mathscr{J}_{I_{\delta},h})\geq \mu'(\mu^{-1}_-(e))+\mathcal{O}(\delta)\,\]
and 
\[ h^{-1}\lambda_{\max}(\mathscr{J}_{I_{\delta},h})\leq \mu'(\mu^{-1}_+(e))+\mathcal{O}(\delta)\,.\]
The converse inequalities can be shown upon considering appropriate trial states. For instance, for $\lambda_{\min}(\mathscr{J}_{I_{\delta},h})$, it is enough to take $\psi_{h}\in E_{h,I_\delta}$ such that
\[\mathscr{F}_h\psi_{h}(x,\xi)=\delta^{-1}(2\pi)^{-\frac{1}{4}}u_{1}(x, \mu^{-1}_-(e)))e^{-\frac{\left(\xi-\mu^{-1}_-(e)\right)^2}{2\delta^4}}\mathds{1}_{I_\delta}(\mu(\xi))\,.\]

\subsection{Proof of Proposition \ref{prop.Agmon}}
Since $e \in (\Theta_0,1)$, we have 
\begin{equation}
\label{eq-ensrec}
\mu_1^{-1}(J_e)=(\mu^{-1}_-(e),\mu^{-1}_+(e))\ \quad\ \mathrm{and}\ \quad \mu_j^{-1}(J_e)=\emptyset,\ j \geq 2. 
\end{equation}
Further, for
$\psi\in E_{J_{e},h}$ we have $\mathds{1}_{J_e}(\mu_1(\xi))\mathscr{F}_h\psi(\cdot,\xi)=\mathscr{F}_h\psi(\cdot,\xi)$ for all $\xi \in \mu_1^{-1}(J_e)$, and consequently 
\[\ell_\xi\mathscr{F}_h\psi(\cdot,\xi)=\mu_1(\xi)\mathscr{F}_h\psi(\cdot,\xi).\]
by arguing as in Lemma \ref{lem.spec-loc}.
Next, by Agmon's theorem (see for instance \cite{Agmon}, \cite[Section 7.2]{FH10} and \cite[Section 4.2]{Raymond}), the following localization formula holds for any $\chi \in W^{1,\infty}(\R_+)$ and all $\xi \in \mu_1^{-1}(J_e)$,
\[\Re\langle \ell_\xi\mathscr{F}_h\psi(\cdot,\xi),\chi^2\mathscr{F}_h\psi(\cdot,\xi)\rangle=q_{\xi}(\chi\mathscr{F}_h\psi(\cdot,\xi))-\|\chi'\mathscr{F}_h\psi(\cdot,\xi)\|^2_{L^2(\R_{+})}\,,\]
where $q_{\xi}$ denotes the quadratic form associated with the operator $\ell_\xi$. It follows that
\[\begin{split}
	\int_{\R_{+}}\left((\xi-x)^2\chi(x)^2-\chi'(x)^2-\chi(x)^2\mu_1(\xi)\right)|\mathscr{F}_h\psi(x,\xi)|^2\dx x &= - \|(\chi\mathscr{F}_h\psi)'(\cdot,\xi)\|^2_{L^2(\R_{+})}\\
	&\leq 0\,.
\end{split}\]
Therefore, since $(\xi-x)^2\geq \frac{x^2}{2}-\xi^2$ and $\mu_1(\xi)\leq e$, we get that
\[\int_{\R_{+}}\left( \left(\frac{x^2}{2}-\xi^2-e\right)\chi(x)^2-\chi'(x)^2\right)|\mathscr{F}_h\psi(x,\xi)|^2\dx x\leq 0\,,\ \xi \in \mu_1^{-1}(J_e).\]
In light of \eqref{eq-ensrec}, this entails that
\[\int_{\R_{+}}\left( \left(\frac{x^2}{2}-C_{e}\right)\chi(x)^2-\chi'(x)^2\right)|\mathscr{F}_h\psi(x,\xi)|^2\dx x\leq 0\,,\]
where $C_e:=\max \{ | \mu^{-1}_-(e)|, | \mu^{-1}_+(e)| \}$.

Thus, taking $\chi(x)=e^{\Phi(x)}$ in the above estimate, where $\Phi$ is a bounded non-negative $K$-Lipschitzian function, we obtain that
\[\int_{\R_{+}}\left(\frac{x^2}{2}-C_{e}-\Phi'^2(x)\right)|e^{\Phi(x)}\mathscr{F}_h\psi(x,\xi)|^2\dx x\leq 0\,,\]
and hence that
\begin{multline*}
\int_{x_{e,K}}^{+\infty} \left( \frac{x^2}{2}-C_{e}-K^2\right)|e^{\Phi(x)}\mathscr{F}_h\psi(x,\xi)|^2\dx x\\
\leq-\int_{0}^{x_{e,K}}\left(\frac{x^2}{2}-C_{e}-K^2\right) |e^{\Phi(x)}\mathscr{F}_h\psi(x,\xi)|^2\dx x \,,
\end{multline*}
where $x_{e,K}:=\sqrt{2(C_e+K^2+1)}>0$.  
Therefore, there exists a constant $C>0$, depending only on $e$, $K$ and $\|\Phi\|_{L^\infty(0,x_{e,K})}$, such that
\[\int_{x_{e,K}}^{+\infty}|e^{\Phi(x)}\mathscr{F}_h\psi(x,\xi)|^2\dx x\leq C\int_{0}^{x_{e,K}} |\mathscr{F}_h\psi(x,\xi)|^2\dx x\,,\]
from where we get
\[\int_{\R_{+}}|e^{\Phi(x)}\mathscr{F}_h\psi(x,\xi)|^2\dx x\leq C\|\mathscr{F}_h\psi(\cdot,\xi)\|^2_{L^2(\R_{+})}\,, \]
upon substituting $C+e^{\|\Phi\|_{L^\infty(0,x_{e,K})}}$ for $C$.
Now, integrating the above estimate with respect to $\xi$ and applying the Parseval formula, we obtain that
\begin{equation}
\label{eq.LiLo}
\int_{\R_{+}^2}|e^{\Phi(x)}\psi(x,y)|^2\dx x\dx y\leq C\|\psi\|^2_{L^2(\R^2_{+})}\,.
\end{equation}
Next, for $n \in \N$, put
\[\Phi_{n}(s):=K \times \begin{cases}
s& \mbox{for } 0\leq s\leq n\,,\\
2n-s& \mbox{for } n\leq s\leq 2n\,,\\
0& \mbox{for } s\geq 2n\,.
\end{cases}\]
Notice that $\|\Phi_n\|_{L^\infty(0,x_{e,K})}=K x_{e,K}$ for all $n \geq x_{e,K}$, in such a way that \eqref{eq.LiLo} holds with $\Phi=\Phi_n$, where the constant $C$ is independent of $n$.

Finally, the result follows from this upon sending $n$ to infinity and applying Fatou's lemma.

\section{Limiting absorption revisited}\label{sec.PAL}
In this section we build a LAP for a self-adjoint operator from a Mourre estimate. The derivation of this result is inspired by \cite{Mourre80} and \cite[Section 4.3]{CFKS87} but we provide here a different approach based on coercivity estimates.

\subsection{Assumptions}
\label{sec-hyp}
We consider two self-adjoint operators $\mathscr{L}$ and $\mathscr{A}$ satisfying
\begin{equation}
\label{eq-ME}
\mathds{1}_J(\mathscr{L})\mathscr{B}\mathds{1}_J(\mathscr{L})\geq c_0\mathds{1}_J(\mathscr{L})\,,\quad \mathscr{B}:=[\mathscr{L},i\mathscr{A}]\,,\quad c_0>0.
\end{equation}
Moreover, we assume that
\begin{enumerate}[\rm (A)]
	\item\label{hyp.i}  $[\mathscr{L},\mathscr{A}](\mathscr{L}+i)^{-1}$ is bounded: 
$$\exists c_1>0,\ \|[\mathscr{L},\mathscr{A}](\mathscr{L}+i)^{-1}\|\leq c_1. $$
	\item\label{hyp.ii} $[[\mathscr{L},\mathscr{A}],\mathscr{A}](\mathscr{L}+i)^{-1}$ is bounded: 
$$\exists c_2>0,\ \|[[\mathscr{L},\mathscr{A}],\mathscr{A}](\mathscr{L}+i)^{-1}\|\leq c_2. $$
\end{enumerate}

\subsection{Limiting absorption through coercivity estimates}
Let $I\subset\subset J$. 
We consider $\eps\geq 0$ and $z \in \C$ such that $\Re z\in I$ and $\Im z\geq 0$. For the sake of notational simplicity we write $z \in I \times [0,+\infty)$ in the sequel. Set 
\[\mathscr{L}_{z,\eps}:=\mathscr{L}-z-i\eps \mathscr{B},\ \mathrm{Dom}\, \mathscr{L}_{z,\eps}:=\mathrm{Dom}\, \mathscr{L}\,.\]
It is apparent that the family $(\mathscr{L}_{z,\eps})_{\eps\in\R}$ is analytic of type (A) in the sense of Kato. Moreover, $\mathscr{L}_{z,\eps}$ is bijective provided $z$ is not on the imaginary axis.

\begin{lemma}
	Let $M>0$. Then, there exists $\varepsilon_0>0$ such that for all $\varepsilon\in(0,\varepsilon_0)$ and all $z\in I\times[M,+\infty)$, the operator $\mathscr{L}_{z,\varepsilon}$ is bijective and satisfies the estimate:
	\[\|\mathscr{L}_{z,\varepsilon}^{-1}\|\leq\frac{4}{\Im z}\leq \frac{4}{M} \,.\]
\end{lemma}
\begin{proof}
	Since
	$\mathscr{L}_{z,\varepsilon}=\mathscr{L}-\Re z-i\Im z-i\varepsilon\mathscr{B}$, we have
	\[\|\mathscr{L}_{z,\varepsilon}u\|\geq\frac12\|(\mathscr{L}-\Re z)u\|+\frac12\Im z\|u\|-\varepsilon\|\mathscr{B}u\|\,,\]
	and hence
	\[2\|\mathscr{L}_{z,\varepsilon}u\|\geq\|(\mathscr{L}-\Re z)u\|+\Im z\|u\|-2c_1\varepsilon\|(\mathscr{L}+i)u\|\,.\]
	It follows that
	\[2\|\mathscr{L}_{z,\varepsilon}u\|\geq\|(\mathscr{L}-\Re z)u\|+\Im z\|u\|-2c_1\varepsilon(\|(\mathscr{L}-\Re z)u\|+\|(i+\Re z)u\|)\,,\]
	and consequently 
	\[2\|\mathscr{L}_{z,\varepsilon}u\|\geq(1-2c_1\varepsilon)\|(\mathscr{L}-\Re z)u\|+\left(\Im z-2c_1\varepsilon\max_{x\in I}|i+x|\right)\|u\|\,.\]
	Next, choosing $\varepsilon$ so small in the above line that $\varepsilon \leq \min \left( \frac{1}{2c_1}, \frac{\max_{x \in I} |i+x|}{4c_1} \right)$, we get that
	\[2\|\mathscr{L}_{z,\varepsilon}u\|\geq\frac{\Im z}{2}\|u\|\,,\]
	which shows that $\mathscr{L}_{z,\varepsilon}$ is injective with closed range. Arguing as above with the adjoint operator $\mathscr{L}^*_{z,\varepsilon}=\mathscr{L}_{\overline{z},-\varepsilon}$ instead of $\mathscr{L}_{z,\varepsilon}$, we get that $\mathscr{L}^*_{z,\varepsilon}$ is injective as well. Therefore, $\mathscr{L}_{z,\varepsilon}$ has a dense range and the conclusion follows.
\end{proof}

\begin{lemma}\label{lem.B1}
	Let $M>0$, set $B:=I\times[0,M]$ and put
	\[\eps_1:=\left(\frac{\sup_{B}|\ell+i|}{\mathrm{dist}(I,J^c)}+1\right)^{-1}\min\left(\frac{c_0}{4c^2_1\sup_{J}|\ell+i|},\frac{1}{2c_1}\right)\,\]
	and
	\[\eps_0:=\min\left(\eps_1,\frac{2\sup_J|\ell+i|}{c_0\left(\frac{\sup_{B}|\ell+i|}{\mathrm{dist}(I,J^c)}+1\right)\left(1+\frac{4c_1}{c_0}\sup_{J}|\ell+i|\right)}\right)\,.\]
	Then, for all $\eps>0$ and all $z\in B$, the following estimates hold.
	\begin{enumerate}[\rm (a)]
		\item\label{eq.b} 
		\begin{multline*}
		\forall u\in\mathrm{Dom}(\mathscr{L})\,,\quad\| \mathds{1}_{J^c}(\mathscr{L})\mathscr{L}_{z,\eps}u\|\geq C(\varepsilon,z)\|(\mathscr{L}+i)\mathds{1}_{J^c}(\mathscr{L})u\|\\
		-c_1\eps\sup_{J}|\ell+i|\|\mathds{1}_{J}(\mathscr{L})u\|\,,
		\end{multline*}
		where
		\[C(\varepsilon,z):=\left(1-c_1\varepsilon\left(1+\frac{|z+i|}{\mathrm{dist}(I,J^c)}\right)\right)\left(1+\frac{|z+i|}{\mathrm{dist}(I,J^c)}\right)^{-1}\,.\]
		Moreover, we have $C(\varepsilon,z) \geq\frac12\left(\frac{\sup_{B}|\ell+i|}{\mathrm{dist}(I,J^c)}+1\right)^{-1}$ for $0<\eps\leq\frac{1}{2c_1}\left(\frac{\sup_{B}|\ell+i|}{\mathrm{dist}(I,J^c)}+1\right)^{-1}$. 
		\item\label{eq.a} 	\begin{equation*}
		\forall u\in\mathrm{Dom}(\mathscr{L})\,,\quad\|\mathscr{L}_{z,\eps}u\|\\	\geq D_1(\varepsilon,z)\|\mathds{1}_J(\mathscr{L})u\|\, 
		\end{equation*}
		where
\[D_1(\varepsilon,z):=\frac{\Im z+c_0\eps-\frac{(c_1\varepsilon)^2\sup_{J}|\ell+i|}{C(\varepsilon,z)}}{1+\frac{c_1\varepsilon}{C(\varepsilon,z)}}\,.\]
Moreover, we have
	$D_1(\epsilon,z) \geq \frac{c_0\eps}{2} \left(1-\frac{c^2_1\varepsilon\sup_{J}|\ell+i|}{c_0C(\varepsilon,z)}\right)
		\geq\frac{c_0\eps}{4}$ provided that
	$0<\eps\leq\eps_1$.
		\item\label{eq.b'}
		\[\forall u\in\mathrm{Dom}(\mathscr{L})\,,\quad\|\mathscr{L}_{z,\eps}u\|\geq D_2(\varepsilon,z)\|(\mathscr{L}+i)\mathds{1}_{J^c}(\mathscr{L})u\|\,,\]
		where
\[D_2(\varepsilon,z)=\frac{C(\varepsilon,z)}{1+c_1\varepsilon D_1(\varepsilon,z)^{-1}\sup_{J}|\ell+i|}\,.\]
Moreover, we have $D_2(\varepsilon,z)\geq\frac{\left(\frac{\sup_{B}|\ell+i|}{\mathrm{dist}(I,J^c)}+1\right)^{-1}}{2\left(1+\frac{4c_1}{c_0}\sup_{J}|\ell+i|\right)}$ whenever $0<\eps\leq\eps_1$.
		\item\label{eq.c} In particular $\mathscr{L}_{z,\eps}$ is bijective and
		\[\|(\mathscr{L}+i)\mathscr{L}_{z,\eps}^{-1}\|\leq D_3(\eps,z)^{-1},\ \]
		where
		\[ D_3(\eps,z):=\frac{1}{\sqrt{2}} \min\left(\frac{D_1(\varepsilon,z)}{\sup_{J}|\ell+i|},D_2(\varepsilon,z)\right)\]
		satisfies $D_3(\eps,z)\geq\frac{c_0\eps}{4\sqrt{2}\sup_{J}|\ell+i|}$ provided that $0<\eps\leq\eps_0$.\\		
		Moreover, for $\eps=0$, $\mathscr{L}_{z,\eps}$ is bijective for $\Im z>0$.
		\item\label{eq.d} \[\forall u\in\mathrm{Dom}(\mathscr{L})\,,\forall \eps \in (0,\eps_2),\quad |\langle \mathscr{L}_{z,\eps}u,u\rangle|+\|\mathscr{L}_{z,\eps}u\|^2\geq \tilde c_0\eps\|u\|^2\,,\]
		where
			\[ \eps_2 :=\min \left( \frac{c_0 \sup_{J}|\ell+i|}{2c_1^2 C(\eps,z)} , \frac{2 D_2(\eps,z)^2}{c_0} , \eps_1\right) \]
and
		\[\tilde c_0:=\frac{c_0}{2}\left(1+\frac{2\left(1+\frac{4c_1}{c_0}\sup_{J}|\ell+i|\right)}{\left(\frac{\sup_{B}|\ell+i|}{\mathrm{dist}(I,J^c)}+1\right)^{-1}}+\frac{4c_1}{c_0} \right)^{-1}.\]	
	\end{enumerate}

\end{lemma} 
\begin{proof}
	\begin{enumerate}[\rm (a)]
		\item We have
		\[\begin{split}\| \mathds{1}_{J^c}(\mathscr{L})\mathscr{L}_{z,\eps}u\|&\geq\|\mathds{1}_{J^c}(\mathscr{L})(\mathscr{L}-z)u\|-\eps\|\mathscr{B} u\|\\
		&\geq \|\mathds{1}_{J^c}(\mathscr{L})(\mathscr{L}-z)u\|-c_1\eps\|(\mathscr{L}+i) u\| \,,
		\end{split}\]
		by \eqref{hyp.i} and thus
				\begin{multline*}\| \mathds{1}_{J^c}(\mathscr{L})\mathscr{L}_{z,\eps}u\|\geq \|\mathds{1}_{J^c}(\mathscr{L})(\mathscr{L}-z)u\|-c_1\eps\sup_{J}|\ell+i|\|\mathds{1}_{J}(\mathscr{L})u\|\\
		-c_1\eps\|(\mathscr{L}+i)\mathds{1}_{J^c}(\mathscr{L})u\|,
		\end{multline*}
		from the orthogonal decomposition of $(\mathscr{L}+i) u$.
		This entails that
		\begin{multline}\label{eq.ap0}
		\| \mathds{1}_{J^c}(\mathscr{L})\mathscr{L}_{z,\eps}u\|\geq (1-c_1\eps)\|\mathds{1}_{J^c}(\mathscr{L})(\mathscr{L}-z)u\|-c_1\eps\sup_{J}|\ell+i|\|\mathds{1}_{J}(\mathscr{L})u\|\\
		-c_1\eps|z+i|\|\mathds{1}_{J^c}(\mathscr{L})u\|\,.
		\end{multline}
		Further, since
		\begin{equation}\label{eq.ap2}
		\|\mathds{1}_{J^c}(\mathscr{L})(\mathscr{L}-z)u\|\geq \mathrm{dist}(I,J^c)\|\mathds{1}_{J^c}(\mathscr{L})u\|\,
		\end{equation}
		and
		\[\|\mathds{1}_{J^c}(\mathscr{L})(\mathscr{L}-z)u\|\geq \|(\mathscr{L}+i)\mathds{1}_{J^c}(\mathscr{L})u\|-|i+z|\|\mathds{1}_{J^c}(\mathscr{L})u\|\,,\]
		we obtain that
		\begin{equation}\label{eq.ap1}
		\|\mathds{1}_{J^c}(\mathscr{L})(\mathscr{L}-z)u\|\geq \left(\frac{|i+z|}{\mathrm{dist}(I,J^c)}+1\right)^{-1} \|(\mathscr{L}+i)\mathds{1}_{J^c}(\mathscr{L})u\|\,. 
		\end{equation}
		Now, plugging \eqref{eq.ap2} and \eqref{eq.ap1} into \eqref{eq.ap0}, we get \eqref{eq.b}.

		\item Since
		\begin{equation}\label{eq.A1}\begin{split}
		&-\Im\langle \mathscr{L}_{z,\eps}u,\mathds{1}_J(\mathscr{L})u\rangle\\
		&=\Im z\|\mathds{1}_J(\mathscr{L})u\|^2+\eps\langle\mathds{1}_J(\mathscr{L})\mathscr{B}\mathds{1}_J(\mathscr{L})u,u\rangle+\eps\Re\langle \mathscr{B}\mathds{1}_{J^c}(\mathscr{L})u,\mathds{1}_J(\mathscr{L})u\rangle\\
		&\geq \Im z\|\mathds{1}_J(\mathscr{L})u\|^2+c_0\eps\|\mathds{1}_J(\mathscr{L})u\|^2-c_1\eps\|\mathds{1}_{J}(\mathscr{L})u\|\|(\mathscr{L}+i)\mathds{1}_{J^c}(\mathscr{L})\|
		\end{split}\end{equation}
		by \eqref{hyp.i} and 
		\[\|(\mathscr{L}+i)\mathds{1}_{J^c}(\mathscr{L})\|\leq C(\varepsilon,z)^{-1}\left(\| \mathds{1}_{J^c}(\mathscr{L})\mathscr{L}_{z,\eps}u\|+c_1\eps\sup_{J}|\ell+i|\|\mathds{1}_{J}(\mathscr{L})u\|\right)\,, \]
		from \eqref{eq.b}, we get that
		\begin{multline}\label{eq.ap3}
		-\Im\langle \mathscr{L}_{z,\eps}u,\mathds{1}_J(\mathscr{L})u\rangle\geq \Im z\|\mathds{1}_J(\mathscr{L})u\|^2+c_0\eps\|\mathds{1}_J(\mathscr{L})u\|^2\\
		-C(\varepsilon,z)^{-1}c_1\varepsilon \|\mathds{1}_J(\mathscr{L})u\|\left(\| \mathds{1}_{J^c}(\mathscr{L})\mathscr{L}_{z,\eps}u\|+c_1\eps\sup_{J}|\ell+i|\|\mathds{1}_{J}(\mathscr{L})u\|\right)\,.
		\end{multline}
		An application of the Cauchy-Schwarz inequality on the left-hand-side of the above inequality then yields
		
		\begin{multline*}
		\left(1+\frac{c_1\varepsilon}{C(\varepsilon,z)}\right)\|\mathscr{L}_{z,\eps}u\|\|\mathds{1}_J(\mathscr{L})u\|\\	\geq \left(\Im z+c_0\eps-\frac{(c_1\varepsilon)^2\sup_{J}|\ell+i|}{C(\varepsilon,z)}\right)\|\mathds{1}_J(\mathscr{L})u\|^2\,,
		\end{multline*}
		which entails \eqref{eq.a}.
		
		\item This statement follows readily from \eqref{eq.b} and \eqref{eq.a}.

		\item We have
		$\|\mathscr{L}_{z,\eps}u\|	\geq 2^{-\frac12} \min(D_1(\varepsilon,z),D_2(\varepsilon,z))\|u\|$ by \eqref{eq.a} and \eqref{eq.b'}, hence $\mathscr{L}_{z,\eps}$ is injective with closed range. Since the same is true for its adjoint $\mathscr{L}_{\overline{z},-\eps}$, the operator $\mathscr{L}_{z,\eps}$ is bijective and $\|\mathscr{L}^{-1}_{z,\eps}\|\leq \frac{\sqrt{2}}{\min(D_1(\varepsilon,z),D_2(\varepsilon,z))}$.
		
		Next, we have
		\begin{equation*}
		\|(\mathscr {L}+i)\mathds{1}_J(\mathscr{L})u\|\leq \sup_{J}|\ell+i|\|\mathds{1}_J(\mathscr{L})u\|\leq D^{-1}_1(\varepsilon,z) \sup_{J}|\ell+i|\|\mathscr{L}_{z,\eps}u\|\,,
		\end{equation*}
		from \eqref{eq.a}. Putting this together with \eqref{eq.b'} we obtain that
		\[\|\mathscr{L}_{z,\eps}u\|\geq \frac{1}{\sqrt{2}} \min\left(\frac{D_1(\varepsilon,z)}{\sup_{J}|\ell+i|},D_2(\varepsilon,z)\right)\|(\mathscr{L}+i)u\|\,.\]
		
		\item By combining the identity
		\[-\Im\langle \mathscr{L}_{z,\eps}u,u\rangle+\Im\langle \mathscr{L}_{z,\eps}u,\mathds{1}_{J^c}(\mathscr{L})u\rangle=-\Im\langle \mathscr{L}_{z,\eps}u,\mathds{1}_{J}(\mathscr{L})u\rangle\,,\]
		with \eqref{eq.ap3}, we get that 
		\begin{multline*}
		-\Im\langle \mathscr{L}_{z,\eps}u,u\rangle+\Im\langle \mathscr{L}_{z,\eps}u,\mathds{1}_{J^c}(\mathscr{L})u\rangle\\
		\geq \left(c_0\eps-\frac{(c_1\varepsilon)^2\sup_{J}|\ell+i|}{C(\varepsilon,z)}\right)\|\mathds{1}_J(\mathscr{L})u\|^2
		-C(\varepsilon,z)^{-1}c_1\eps\|\mathds{1}_{J}(\mathscr{L})u\|\|\mathscr{L}_{z,\eps} u\|\,,
		\end{multline*}
		and consequently
		\begin{multline*}
		|\langle \mathscr{L}_{z,\eps}u,u\rangle|+\|\mathscr{L}_{z,\eps}u\|\|\mathds{1}_{J^c}(\mathscr{L})u\|+C(\varepsilon,z)^{-1}c_1\eps\|\mathds{1}_{J}(\mathscr{L})u\|\|\mathscr{L}_{z,\eps} u\|\\
		\geq \left(c_0\eps-\frac{(c_1\varepsilon)^2\sup_{J}|\ell+i|}{C(\varepsilon,z)}\right)\|\mathds{1}_{J}(\mathscr{L})u\|^2\,.\end{multline*}
		From this, \eqref{eq.a} and \eqref{eq.b'}, it then follows that
		\begin{multline*}
		|\langle \mathscr{L}_{z,\eps}u,u\rangle|+\left(D_2(\varepsilon,z)^{-1}+D_1(\varepsilon,z)^{-1}C(\varepsilon,z) c_1\varepsilon \right)\|\mathscr{L}_{z,\eps}u\|^2\\
		\geq \left(c_0\eps-\frac{(c_1\varepsilon)^2\sup_{J}|\ell+i|}{C(\varepsilon,z)}\right)\|\mathds{1}_{J}(\mathscr{L})u\|^2\,.		
		\end{multline*}
		Next, with reference to \eqref{eq.b'} we may add $\|\mathscr{L}_{z,\eps}u\|^2$ on the left-hand-side of the above estimate and $D_2(\varepsilon,z)^2\|\mathds{1}_{J^c}(\mathscr{L})u\|^2$ on its right-hand-side. We obtain that
		\begin{multline*}
		|\langle \mathscr{L}_{z,\eps}u,u\rangle|+\left(1+D_2(\varepsilon,z)^{-1}+D_1(\varepsilon,z)^{-1}C(\varepsilon,z) c_1\varepsilon \right)\|\mathscr{L}_{z,\eps}u\|^2\\
		\geq \left(c_0\eps-\frac{(c_1\varepsilon)^2\sup_{J}|\ell+i|}{C(\varepsilon,z)}\right)\|\mathds{1}_{J}(\mathscr{L})u\|^2+D_2(\varepsilon,z)^2\|\mathds{1}_{J^c}(\mathscr{L})u\|^2\,.		
		\end{multline*}
		As a consequence we have for all $\eps\in(0,\eps_2]$,
		\begin{equation*}
		|\langle \mathscr{L}_{z,\eps}u,u\rangle|+\left(1+D_2(\varepsilon,z)^{-1}+D_1(\varepsilon,z)^{-1}C(\varepsilon,z) c_1\varepsilon \right)\|\mathscr{L}_{z,\eps}u\|^2\\
		\geq\frac{c_0\eps}{2}\|u\|^2\,.		
		\end{equation*}
		Bearing in mind that $C(\epsilon,z) \in (0,1)$, this entails that
		\begin{equation*}
		|\langle \mathscr{L}_{z,\eps}u,u\rangle|+\left(1+\frac{2\left(1+\frac{4c_1}{c_0}\sup_{J}|\ell+i|\right)}{\left(\frac{\sup_{B}|\ell+i|}{\mathrm{dist}(I,J^c)}+1\right)^{-1}}+\frac{4c_1}{c_0} \right)\|\mathscr{L}_{z,\eps}u\|^2\\
		\geq\frac{c_0\eps}{2}\|u\|^2\,,		
		\end{equation*}
		which yields the desired result. 
		
	\end{enumerate}
\end{proof}

Having established Lemma \ref{lem.B1}, we can now state the following technical result.

\begin{lemma}\label{lem.B2}
	For all bounded self-adjoint operator $\mathscr{C}$, we have
	\[\|\mathscr{L}_{z,\eps}^{-1}\mathscr{C}\|\leq \left(\frac{1}{\tilde c_0\eps}\right)^{\frac12}(\|\mathscr{C} \| +\|\mathscr{C}\mathscr{L}_{z,\eps}^{-1}\mathscr{C}\|^{\frac 12})\,\]
	and
	\[\|\mathscr{C}\mathscr{L}_{z,\eps}^{-1}\|\leq \left(\frac{1}{\tilde c_0\eps}\right)^{\frac12}(\|\mathscr{C} \|+\|\mathscr{C}\mathscr{L}_{z,\eps}^{-1}\mathscr{C}\|^{\frac 12})\,.\]
\end{lemma}
\begin{proof}
	Taking $u=\mathscr{L}_{z,\eps}^{-1}\mathscr{C}\varphi$ in \eqref{eq.d}, we obtain the first estimate. The second one follows from this and the fact that $\mathscr{C}\mathscr{L}_{z,\eps}^{-1}$ is the adjoint of $\mathscr{L}_{\overline{z},-\eps}^{-1}\mathscr{C}$.
\end{proof}
Armed with Lemma \ref{lem.B2}, we are in position to prove the main result of this section.

\begin{theorem}[Limiting Absorption Principle]\label{prop.B3}
Let $\mathscr{L}$ and $\mathscr{A}$ fulfill the conditions of Section \ref{sec-hyp}. Then, for any bounded self-adjoint operator $\mathscr{C}$ such that $\mathscr{C}\mathscr{A}$ and $\mathscr{A}\mathscr{C}$ are bounded, it holds true for all $\eps\in(0,\min(1,\eps_0)]$ that
\[ \sup_{\Im z> 0\,,\Re z\in I}\|\mathscr{C}(\mathscr{L}-z-i\eps \mathscr{B})^{-1}\mathscr{C}\| \le C, \]
	where
	\[ C:=C(\eps_0)+(K_1+K_2)(1+C(\eps_0)^{\frac12})\int_0^1\frac{\mathrm{d}t}{t^{\frac12}}+\sqrt{2}K^{\frac12}(K_1+K_2)\int_0^1\frac{|\ln(t)|^{\frac12}}{t^{\frac12}}\mathrm{d}t, \]
	\[K_1:=\frac{2}{\sqrt{\tilde c_0}}\max(\|\mathscr{C}\mathscr{A}\|,\|\mathscr{A}\mathscr{C}\|)\,,\quad K_2:=\frac{4 \sqrt{2} c_2\sup_{J}|\ell+i|}{c_0}\|\mathscr{C}\|\,,\]
	\[K:=\left(K_1+	K_2\right)\|\mathscr{C}\|\left(1+\frac{2\sup_{J}|\ell+i|^{\frac12}}{\sqrt{c_0}}\right)\,\]
	and $C(\eps_0)$ is a positive constant satisfying
	\[C(\eps_0)\leq \frac{4 \sqrt{2} \sup_{J}|\ell+i|}{c_0}\|\mathscr{C}\|^2\eps_0^{-1}\,.\]
	Moreover, we have
	\[\sup_{\Im z> 0\,,\Re z\in I}\|\mathscr{C}(\mathscr{L}-z)^{-1}\mathscr{C}\|\leq C\,.\]
\end{theorem}
\begin{proof}
	Let us differentiate $F(\eps):=\mathscr{C}\mathscr{L}_{z,\eps}^{-1}\mathscr{C}$ w.r.t. $\eps$. We obtain that
	\[\begin{split}F'(\eps)&=\mathscr{C}\mathscr{L}_{z,\eps}^{-1}[\mathscr{L},\mathscr{A}]\mathscr{L}_{z,\eps}^{-1}\mathscr{C}\\
	&=\mathscr{C}\mathscr{L}_{z,\eps}^{-1}[\mathscr{L}_{z,\eps},\mathscr{A}]\mathscr{L}_{z,\eps}^{-1}\mathscr{C}-\eps \mathscr{C}\mathscr{L}_{z,\eps}^{-1}[[\mathscr{L},\mathscr{A}],\mathscr{A}]\mathscr{L}_{z,\eps}^{-1}\mathscr{C}\\
	&= \mathscr{C}\mathscr{A} \mathscr{L}_{z,\eps}^{-1}\mathscr{C} - \mathscr{C}\mathscr{L}_{z,\eps}^{-1} \mathscr{A}\mathscr{C}-\eps \mathscr{C}\mathscr{L}_{z,\eps}^{-1}[[\mathscr{L},\mathscr{A}],\mathscr{A}]\mathscr{L}_{z,\eps}^{-1}\mathscr{C}.\end{split}\]
	Bearing in mind that $\mathscr{A}\mathscr{C}$ and $\mathscr{C}\mathscr{A}$ are bounded, we refer to \eqref{hyp.ii}, \eqref{eq.c} and Lemma \ref{lem.B2}, and deduce from the above estimate that
	\begin{equation} \label{ineq.diff}
	\begin{split}\|F'(\eps)\|&\leq \|\mathscr{C}\mathscr{A}\|\|\mathscr{L}^{-1}_{z,\eps}\mathscr{C}\|+\|\mathscr{A}\mathscr{C}\|\|\mathscr{C}\mathscr{L}^{-1}_{z,\eps}\|+c_2\eps\|\mathscr{C}\|\|\mathscr{C}\mathscr{L}^{-1}_{z,\eps}\|\|(\mathscr{L}+i)\mathscr{L}^{-1}_{z,\eps}\|\\
	&	\leq \left(\frac{2}{\sqrt{\tilde c_0\eps}}\max(\|\mathscr{C}\mathscr{A}\|,\|\mathscr{A}\mathscr{C}\|)+c_2\eps\|\mathscr{C}\|D_3(\eps,z)^{-1}\right)(\|\mathscr{C}\|+\|F\|^{\frac 12})\\
	&\leq \left(K_1\eps^{-\frac12}+	K_2\right)(\|\mathscr{C}\|+\|F\|^{\frac 12})\,. 
	\end{split}
	\end{equation}
	Further, since
	$\|F(\eps)\|\leq\|\mathscr{C}\|^2D_3(\eps,z)^{-1}\leq \frac{4\sup_{J}|\ell+i|}{c_0}\|\mathscr{C}\|^2\eps^{-1}$ for all $\eps\in(0,\eps_0]$, by \eqref{eq.c}, we infer from \eqref{ineq.diff} upon possibly substituting $1$ for $\eps_0$, that
		\[\begin{split}
	\|F'(\eps)\|&\leq\left(K_1\eps^{-\frac12}+	K_2\right)\|\mathscr{C}\|\left(1+\eps^{-\frac12}\frac{2\sup_{J}|\ell+i|^{\frac12}}{\sqrt{c_0}}\right)\\
	&\leq \left(K_1+	K_2\right)\|\mathscr{C}\|\left(1+\frac{2\sup_{J}|\ell+i|^{\frac12}}{\sqrt{c_0}}\right)\eps^{-1}.
	\end{split}
	\]
	Integrating the above estimate over $(\eps,\eps_0)$ then yields
\[\|F(\eps)\|\leq K|\ln(\eps)|+\|F(\eps_0)\|\,, \eps \in (0,\eps_0]. \]
	Plugging this into \eqref{ineq.diff}, we obtain that
	\[\|F'(\eps)\|\leq (K_1+K_2)\eps^{-\frac12}\left(\|\mathscr{C}\|+K^{\frac12}|\ln(\eps)|^{\frac12}+\|F(\eps_0)\|^{\frac12}\right)\,,\]
	which, upon integrating over $(\eps,\eps_0)$, leads to
	\[ \| F(\eps) \| \le  \|F(\eps_0) \| + (K_1+K_2) \left( \left( \|\mathscr{C}\| + \|F(\eps_0)\|^{\frac12} \right) \int_0^1 \frac{\mathrm{d}t}{t^{\frac12}} +K^{\frac12} \int_0^1 \frac{|\ln t|^{\frac12}}{t^{\frac12}} \mathrm{d} t \right). \]
	This and the estimate
	\begin{equation*}
	\begin{split}
	\|F(\eps_0) \| & = \| \mathscr{C} (\mathscr{L}+i)^{-1} (\mathscr{L}+i) \mathscr{L}^{-1}_{z,\eps_0} \mathscr{C} \| \\
	& \le \| (\mathscr{L}+i) \mathscr{L}^{-1}_{z,\eps_0} \| \| \mathscr{C} \|^2\\
	& \le \frac{4\sqrt{2}\sup_{J}|\ell+i|}{c_0} \| \mathscr{C} \|^2 \eps_0^{-1},
	\end{split}
	\end{equation*}
	arising from \eqref{eq.c}, yield the desired result with $C(\eps_0)=\|F(\eps_0) \|$.
	
\end{proof}

\section{Mourre estimates and limiting absorption}\label{sec.Mourre}
\subsection{A Mourre estimate for the unperturbed operator}

Since $V(x,\cdot,h) \in S_{\frac\alpha 2}(\langle y\rangle^{-1})$ and $\mathscr{A}_h\in S_{\frac\alpha 2}(\langle y\rangle)$, the following lemma is a direct consequence of the composition theorem of pseudo-differential operators (see \cite[Theorem 4.18]{Z13}) and the Calderon-Vaillancourt theorem (see \cite[Theorem 4.23]{Z13}).

\begin{lemma}\label{lem.commutVA}
	The pseudo-differential operator $[V(x,\cdot,h),\mathscr{A}_h]$ is bounded from $L^2(\R_y)$ to $L^2(\R_y)$. More precisely, there exist $C>0$ and $h_0>0$ such that for all $h\in(0,h_0)$ and all $x>0$, we have
	\[\|[V(x,\cdot,h),\mathscr{A}_h]\|_{\mathscr{L}^2(L^2(\R_y))}\leq Ch^{1-\alpha }\,.\]	
	In particular, this operator extends to a bounded operator on $L^2(\R_+^2)$.
\end{lemma}

\begin{remark}
The expression of the above commutator can be explicitly calculated at the cost of rather tedious computations, but it is not needed thanks to the pseudo-differential approach.
\end{remark}

\begin{lemma}\label{lem.comutLA}
	The following commutator is well defined on $\mathrm{Dom}(\mathscr{L}_h)$ and
	\[[\mathscr{L}_{h},i\mathscr{A}_h]=4h(hD_y-x)f(hD_y)\,.\]
\end{lemma}

\begin{proof}
	Since
	\begin{equation}\label{eq.commutator}
	\begin{split}
	[\mathscr{L}_{h},i\mathscr{A}_h]&=[(hD_y-x)^2,i\mathscr{A}_h]=i[(hD_y-x)^2,yf(hD_y)+f(hD_y)y]\\
	&=i[(hD_y-x)^2,y]f(hD_y)+ihf(hD_y)[(hD_y-x)^2,y],\\
	\end{split}	
	\end{equation}
we get the desired result upon recalling that $[(hD_y-x)^2,y]=-2ih(hD_y-x)$.
\end{proof}

For $d:=bh^\alpha$, put $J:=[e-d,e+d]$. Since $\beta > 2 \alpha$, we have $0<d<\delta$ whenever $h \in (0,1)$, and consequently $I$ is a proper subset of $J$: $\overline{I} \subset J$.

\begin{proposition} \label{pr-MEunperturbed}
	There exist $\tilde{c}_0>0$ and $h_0>0$ such that for all $h\in(0,h_0)$, we have
	\[\mathds{1}_{J}(\mathscr{L}_{h})[\mathscr{L}_{h},i\mathscr{A}_h]\mathds{1}_{J}(\mathscr{L}_{h})\geq \tilde{c}_0 h^{1+\alpha}>0\,.\]	
\end{proposition}

\begin{proof}
	By the Parseval formula, we get from Lemma \ref{lem.comutLA} that for all $u\in \mathrm{Ran}(\mathds{1}_J(\mathscr{L}_h))$,
	\[\langle[\mathscr{L}_{h},i\mathscr{A}_h]u,u\rangle_{L^2(\R_+\times\R)}=4h\langle(\eta-x)f(\eta)\mathscr{F}_h u,\mathscr{F}_h u\rangle_{L^2(\R_+\times\R_\eta)}\,.\]
	Moreover, analogously to Lemma \ref{lem.spec-loc} we have
	\[\mathscr{F}_h u(x,\eta)=\mathds{1}_{\mu_1(\eta)\in J}U(\eta) u_1(x,\eta)\,,\]
	where
	\[U(\eta)=\langle\mathscr{F}_h u (\cdot,\eta),u_1(\cdot,\eta)\rangle_{L^2(\mathbb{R}_+)}\,,\]
	in such a way that
	\[\langle[\mathscr{L}_{h},i\mathscr{A}_h]u,u\rangle_{L^2(\R_+\times\R)}=4h\int_{\mathbb{R}_+}\int_{\R}(\eta-x)f(\eta)\mathds{1}_{\mu_1(\eta)\in J}|U(\eta)|^2| u_1(x,\eta)|^2\mathrm{d}\eta\mathrm{d}x\,.\]
	Therefore, applying the Feynman-Hellmann formula, we obtain that
	\[\langle[\mathscr{L}_{h},i\mathscr{A}_h]u,u\rangle_{L^2(\R_+\times\R)}=2h\int_{\mu_1(\eta)\in J}f(\eta)\mu'_1(\eta)| U(\eta)|^2\mathrm{d}\eta\,.
	\]
	Moreover, we have
	$\int_{\mu_1(\eta)\in J}f(\eta)\mu'_1(\eta)| U(\eta)|^2\mathrm{d}\eta\geq\int_{\mu_1(\eta)\in J}|\mu'_1(\eta)|| U(\eta)|^2\mathrm{d}\eta$ from the definition of $f$, and consequently 
		\[\langle[\mathscr{L}_{h},i\mathscr{A}_h]u,u\rangle_{L^2(\R_+\times\R)}\geq \tilde{c}_0 h^{1+\alpha}\int_{\mu_1(\eta)\in J}|U(\eta)|^2\mathrm{d}\eta \geq \tilde{c}_0 h^{1+\alpha}\|u\|^2\,,\]
		from the quadratic behavior of $\mu_1$ at its minimum, expressed in Proposition \ref{prop.rappel}.iii. 

\end{proof}

One of the benefits of a Mourre estimate is its stability under perturbation. Having established a Mourre inequality for the unperturbed operator $\mathscr{L}_{h}$ in the above proposition, we turn now to extending this estimate to the case of the perturbed operator $\mathscr{L}_{h}+V$.

\subsection{The case of the perturbed operator}
Since the proof of a Mourre estimate for the perturbed operator essentially boils down to the existence of a Mourre inequality for unperturbed operator, we preliminarily establish the following property of the spectral decomposition associated with $\mathscr{L}_h$.

\subsubsection{Spectral decomposition associated with $\mathscr{L}_h$}
\begin{lemma}\label{lem.orth.dec}
	Let $\phi\in\mathrm{range}\,\mathds{1}_I(\mathscr{L}_{h,V})$. Then, $\phi$ decomposes as
	\[\phi=\phi_1+\phi_2\,,\qquad \phi_1=\mathds{1}_J(\mathscr{L}_h)\phi\,,\quad \phi_2=\mathds{1}_{\R\setminus J}(\mathscr{L}_h)\phi\,,\]
	and we have
	\[\|\phi_2\|\leq c_h \|\phi\|\,,\quad c_h =d^{-1}(\delta+h^\gamma\|V\|_\infty)\,.\]
\end{lemma}

\begin{proof}
	Since $\phi_2=\mathds{1}_{\R\setminus J}(\mathscr{L}_h)\mathds{1}_I(\mathscr{L}_{h,V})\phi$, we have
	\[\begin{split}
	\phi_2&=(\mathscr{L}_h-e)^{-1}\mathds{1}_{\R\setminus J}(\mathscr{L}_h)(\mathscr{L}_h-e)\mathds{1}_I(\mathscr{L}_{h,V})\phi\\
	&=(\mathscr{L}_h-e)^{-1}\mathds{1}_{\R\setminus J}(\mathscr{L}_h)(\mathscr{L}_{h,V}-e-h^\gamma V)\mathds{1}_I(\mathscr{L}_{h,V})\phi\,,
	\end{split}
	\]
	which immediately yields that $\|\phi_2\|\leq d^{-1}(\delta+h^\gamma\|V\|_\infty)\|\phi\|$.
\end{proof}

\subsubsection{Proof of Theorem \ref{theo.Mourre}}
For all $\phi\in\mathrm{range}\, \mathds{1}_I(\mathscr{L}_{h,V})$, we have
\begin{equation}\label{eq.Qh} 
\langle[\mathscr{L}_{h,V},i\mathscr{A}_h]\phi,\phi \rangle=\langle[\mathscr{L}_{h},i\mathscr{A}_h]\phi,\phi \rangle+h^\gamma\langle[V,i\mathscr{A}_h]\phi,\phi \rangle\,.
\end{equation}
The second term on the right-hand-side of the above line can be estimated with the aid of
Lemma \ref{lem.commutVA}: since $\gamma-\alpha+1>1+\alpha$, we get that
\begin{equation}\label{eq.pt}
h^\gamma\left|\langle[V,i\mathscr{A}_h]\phi,\phi \rangle\right|\leq h^\gamma\|[V,i\mathscr{A}_h]\|_{\mathscr{L}(L^2(\R^2_+))}\|\phi\|^2\leq Ch^{\gamma-\alpha+1}\|\phi\|^2\leq \frac{\tilde{c}_0 }{4}h^{1+\alpha}\|\phi\|^2\,,
\end{equation}
provided that $h$ is small enough.
In the first term on right-hand-side of \eqref{eq.Qh}, the commutator $[\mathscr{L}_{h},i\mathscr{A}_h]$ is acting on the state $\phi\in\mathrm{range}\,\mathds{1}_I(\mathscr{L}_{h,V})$, which decomposes according to Lemma \ref{lem.orth.dec}, giving:
\[\langle[\mathscr{L}_{h},i\mathscr{A}_h]\phi,\phi \rangle=\langle[\mathscr{L}_{h},i\mathscr{A}_h]\phi_1,\phi_1 \rangle+\langle[\mathscr{L}_{h},i\mathscr{A}_h]\phi_2,\phi_2 \rangle+2\mathrm{Re}\langle[\mathscr{L}_{h},i\mathscr{A}_h]\phi_1,\phi_2 \rangle\,.\]
Therefore, we have
\[\langle[\mathscr{L}_{h},i\mathscr{A}_h]\phi,\phi \rangle\geq \tilde{c}_0 h^{1+\alpha}\|\phi_1\|^2-|\langle[\mathscr{L}_{h},i\mathscr{A}_h]\phi_2,\phi_2 \rangle|-2|\langle[\mathscr{L}_{h},i\mathscr{A}_h]\phi_1,\phi_2 \rangle|\,,\]
by Proposition \ref{pr-MEunperturbed}, and hence
\begin{multline*}
\langle[\mathscr{L}_{h},i\mathscr{A}_h]\phi,\phi \rangle\geq \tilde{c}_0 h^{1+\alpha}\|\phi_1\|^2-4h\|f\|_\infty|\|\phi_2\|\|(hD_y-x)\phi_2\|\\
-8h\|f\|_\infty|\|(hD_y-x)\phi_1\|\|\phi_2\|\,,
\end{multline*}
from Lemma \ref{lem.comutLA}. Further, taking into account that 
\[\|(hD_y-x)\phi_1\|^2\leq \|\phi_1\|\|\mathscr{L}_h\phi_1\|\leq (e+d)\|\phi_1\|^2\,\]
and that
\[\begin{split}
\|(hD_y-x)\phi_2\|^2= \langle\phi_2,\mathscr{L}_h\phi_2\rangle=\langle\phi_2,\mathscr{L}_h\phi\rangle&=\langle\phi_2,(\mathscr{L}_h+h^\gamma V)\phi\rangle-\langle\phi_2,h^\gamma V\phi\rangle\\
&\leq (e+\delta+h^\gamma\|V\|_\infty)\|\phi_2\|\|\phi\|\,,
\end{split} 
\]
we get that
\begin{multline*}
\langle[\mathscr{L}_{h},i\mathscr{A}_h]\phi,\phi \rangle\geq \tilde{c}_0 h^{1+\alpha}\|\phi_1\|^2\\
-4h(e+\delta+h^\gamma\|V\|_\infty)^{\frac 12}\|f\|_\infty\|\phi_2\|^{\frac 32}\|\phi\|^{\frac 12}-8h(e+d)^{\frac 12}\|f\|_\infty\|\phi_1\|\|\phi_2\|\,.
\end{multline*}
From this and Lemma \ref{lem.orth.dec} it then follows that
\begin{equation*}
\langle[\mathscr{L}_{h},i\mathscr{A}_h]\phi,\phi \rangle\geq \left( \tilde{c}_0 h^{1+\alpha} - \left( 
8 (e+d+ h^\gamma \|V\|_\infty)^{\frac12} \|f\|_\infty (1+ c_h^{\frac 12}) + c_h h^\alpha \right) c_h h \right)  \|\phi\|^2\,.
\end{equation*}
Thus, bearing in mind that $d=h^\alpha$ in such a way that $c_h$ (which scales like $h^{\beta-\alpha}$) can be made arbitrarily small relative to $h^\alpha$ (as we have $\beta >2 \alpha$) in the asymptotic regime $h \downarrow 0$, we infer from the above estimate that
\[
\langle[\mathscr{L}_{h},i\mathscr{A}_h]\phi,\phi \rangle\geq \frac{3\tilde{c}_0}{4} h^{1+\alpha}\|\phi\|^2\,,
\]
whenever $h$ is sufficiently small.
Putting this together with \eqref{eq.Qh}-\eqref{eq.pt}, we get the result of Theorem \ref{theo.Mourre} upon replacing $\frac{\tilde{c}_0}{2}$ by $\tilde{c}_0$.

\subsection{Proof of Corollary \ref{cor.PAL}}
By Theorem \ref{theo.Mourre}, the self-adjoint operator $\mathscr{L}_{h,V}$ satisfies a Mourre estimate of type \eqref{eq-ME} on the interval $I$, associated with the local commutator $\mathscr{A}_h$ and a constant $c_0$ of size $h^{1+\alpha}$. 
Next, since $f(hD_y)$ is a bounded operator by definition of $f$, Lemmas \ref{lem.commutVA} and \ref{lem.comutLA} ensure us that the pair $(\mathscr{L}_{h,V},\mathscr{A}_h)$ fulfills the assumption \eqref{hyp.i} of Section \ref{sec-hyp} with a constant $c_1$ of order $\mathscr{O}(h^{1-\alpha})$. Moreover, we have
\[\begin{split}
[[\mathscr{L}_h,\mathscr{A}_h],\mathscr{A}_h]&=-4ih[(hD_y-x)f(hD_y),yf_h(hD_y)+f_h(hD_y)y]\\
&=-4ih\left(-2ih(hD_y-x)f'(hD_y)f(hD_y)-2ihf^2(hD_y)\right)\,,
\end{split}\]
by a straightforward computation. We deduce from this and the boundedness of the operators $f(hD_y)$ and $f'(hD_y)$ that 
$[[\mathscr{L}_h,\mathscr{A}_h],\mathscr{A}_h](\mathscr{L}_h+i)^{-1}$ is bounded and that its norm is of order $\mathscr{O}(h^{2-\alpha})$. Moreover, since $V(x,\cdot,h) \in S_{\frac\alpha 2}(\langle y\rangle^{-1})$ by assumption and $\mathscr{A}_h\in S_{\frac\alpha 2}(\langle y\rangle)$, we can check upon arguing in the same fashion as in the derivation of Lemma \ref{lem.commutVA} that $(\mathscr{L}_{h,V},\mathscr{A}_h)$ satisfies the assumption \eqref{hyp.ii} of
Section \ref{sec-hyp}, where the constant $c_2$ is of order $\mathscr{O}(h^{2-\alpha})$. 

Therefore, with reference to Section \ref{sec.PAL}, the statement of Corollary \ref{cor.PAL} follows directly from Theorem \ref{prop.B3} and the fact that we can actually track the powers of $h$ in Lemma \ref{lem.B1} and Theorem \ref{prop.B3}. As a matter of fact it can be checked that 
\[\left(\mathrm{dist}(I,J^c)\right)^{-1}\simeq h^{-\alpha}\,, \quad\ \eps_1\simeq h^{-1+4\alpha},\ \quad\ \eps_0\simeq h^{-1+4\alpha}\ \quad\ \mathrm{and}\ \quad \tilde c_0\gtrsim h^{1+\alpha}\]
in Lemma \ref{lem.B1}, whereas in Theorem \ref{prop.B3} we have 
\[K_1=\mathscr{O}(h^{-\frac12-\frac{\alpha}{2}}),\ \quad\ K_2=\mathscr{O}(h^{1-2\alpha}),\ \quad\ K=\mathscr{O}(h^{-1-\alpha}),\ \quad\ C(\eps_0)=\mathscr{O}(h^{-2+3\alpha})\]
and finally 
\[C=\mathscr{O}(h^{\min(-2+3\alpha,-3/2+\alpha,-1-\alpha)})\,.\]

\subsection*{Acknowledgments}
Nicolas Raymond and \'Eric Soccorsi are deeply grateful to the CIRM where this work was completed.
\'Eric Soccorsi is partially supported by the Agence Nationale de la Recherche (ANR) under grant ANR-17-CE40-0029.

\bibliographystyle{abbrv}
\bibliography{bib-RS3}
\end{document}